\documentclass[times,numbers,compress]{article}

\usepackage{graphicx}
\usepackage{epstopdf}
\usepackage{amsfonts}
\usepackage{color}
\usepackage{tikz}
\usetikzlibrary{arrows.meta,hobby,decorations.markings}
\usepackage{mathrsfs}
\usepackage{psfrag, amssymb, amsmath, cite}
\usepackage{verbatim}
\usepackage{epsfig} 
\usepackage{times} 
\usepackage{amsmath} 
\usepackage{graphicx}

\usepackage{amssymb}  
\usepackage{amsfonts}
\usepackage[colorlinks,linkcolor=red,anchorcolor=blue,citecolor=blue]{hyperref}
\usepackage{cite}
\usepackage[ruled]{algorithm2e}
\usepackage{fancybox}
\usepackage{framed}
\usepackage{amsthm}
\usepackage{mathtools}

\DeclarePairedDelimiter\floor{\lfloor}{\rfloor}
\newtheorem{remark}{Remark}
\newtheorem{definition}{\textbf{Definition}}

\newtheorem{theorem}{\textbf{Theorem}}

\newtheorem{lemma}{\textbf{Lemma}}

\newtheorem{example}{\textbf{Example}}

\usepackage{fancyhdr}
\newcommand\BibTeX{{\rmfamily B\kern-.05em \textsc{i\kern-.025em b}\kern-.08em
T\kern-.1667em\lower.7ex\hbox{E}\kern-.125emX}}

\def\volumeyear{2018}

\begin{document}

\title{Quantization effects and convergence properties of rigid formation control systems with quantized distance measurements}

\author{Zhiyong Sun$^1$, Hector Garcia de Marina$^2$, \\ Brian D. O. Anderson$^{1,3}$ and Ming Cao$^4$ \\ {\small
$^1$ Data61-CSIRO (formerly NICTA) and
Research School of Engineering,} \\
{\small The Australian National University, Canberra ACT 2601, Australia.} \\
{\small $^2$ Unmanned Aerial Systems Centre, The Maersk Mc-Kinney Moller Institute.} \\{\small Southern University of Denmark.} \\
{\small $^3$ School of Automation, Hangzhou Dianzi University, Hangzhou 310018, China.} \\
{\small $^4$ Faculty of Mathematics and Natural Sciences,} \\ {\small ENTEG, University of Groningen, The Netherlands.}}




\maketitle
\section*{Abstract}

In this paper we discuss  quantization effects in rigid formation control systems when target formations are described by inter-agent distances.
Because of practical sensing and measurement constraints, we consider in this paper  distance measurements in their quantized forms. We show that under gradient-based formation control, in the case of uniform quantization, the distance errors converge locally to a bounded set whose size depends on the quantization error, while in the case of logarithmic quantization, all distance errors converge locally to zero. A special quantizer involving the signum function is then considered with which all agents can only measure coarse distances in terms of binary information. In this case the formation converges locally to a target formation   within a  finite time. Lastly, we   discuss the effect of asymmetric uniform quantization on rigid formation control.

\section{Introduction}

Quantized control has been an active research topic in the recent decade, motivated by the fact that digital sensors and numerous industrial controllers  can only generate quantized   measurements or  feedback signals \cite{brockett2000quantized, liberzon2003hybrid}.   Recent years have also witnessed extensive discussions on   quantized control for  networked control systems. This is because   data exchange and transmission  over networks often occurs  in a digitally quantized manner, thus giving rise to coarse and imperfect information; see e.g., \cite{kashyap2007quantized,   cai2011quantized, ceragioli2011discontinuities, liu2012quantization,  frasca2012continuous, guo2013consensus}.

In this paper, we aim to discuss the quantization effect on rigid formation control. Formation control based on graph rigidity is a typical networked control problem involving inter-agent measurements and cooperations. There have been many papers in the literature focusing  on control performance and convergence analysis for rigid formation control systems (see e.g. \cite{krick2009stabilisation, dorfler2010geometric, oh2014distance, oh2015survey,sun2015rigid, sun2016exponential}), with virtually all assuming that all agents can acquire the relative position measurements to their neighbors perfectly. We remark  that there are some recent works on linear-consensus-based formation control with quantized measurements. An exemplary  paper  along this line of research is \cite{Jafarian2015125}, which showed that by using very coarse measurements (i.e., measurements in terms of binary information) the formation stabilization task  can still be achieved. The case of coarse measurements can be seen as a special (or extreme) quantizer, which generates quantized feedbacks in the form of binary signals. However,  in \cite{Jafarian2015125} and similar works on linear-consensus-based formation control, a common knowledge of the global coordinate frame orientation is required for all the agents to implement the control law. This is a  strict assumption and is not always desirable  in practical formation control systems. Actually, it has been shown in \cite{meng2016formation}  that coordinate orientation mismatch may also cause undesired formation motions in linear-consensus-based formation systems. All these restrictions and disadvantages  are known to be avoidable  in rigid formation control systems, in which any common  knowledge of the global coordinate system is not required.

In the framework of quantized formation control, we also consider in the latter part of this paper a special quantizer  described by the \emph{signum} function. This part is motivated by the previous work \cite{liu2014controlling} which discussed  \emph{triangular} formation control with coarse distance measurements involving the signum function. In this paper we will consider a more general setting, which extends the discussions from  the triangular case in \cite{liu2014controlling} to more general rigid formations.

The  aim of this paper is to explore whether the introduction of quantized measurement and feedback can still guarantee the success  of formation control, and to what extent the controller performance limits exist. Our broad conclusion is that quantization is not fatal, but may reduce performance in achieving a target formation. {\color{black}Generally speaking, quantized measurements are a type of approximation of actual measurements, and such approximations bring about bounded measurement errors that depend on quantizer functions. We remark that formation shape control with distance measurement errors or biases is discussed in \cite{de2018taming, Mou_Problematic, de2015controlling, sun2017rigid}.  Measurement errors due to quantizations are different to measurement noises, in the sense that measurement errors induced by quantizations are deterministic, and some quantizers (especially logarithmic quantizers and binary quantizers) can also distinguish whether the quantity under quantization (distance error in the context of formation control) is zero.   
In this respect,  measurements might be coarse from quantization,
but the most important information (distance error being zero or not)
is known without any noise. Furthermore, if quantization errors are not large and do not exceed the convergence region for a formation control system, then the local stability and convergence of  formation systems are still guaranteed.
For different types of quantizers applied in a distance measurement, we will present a detailed analysis to characterize the actual convergence  and the boundedness of formation errors. 

}  

In this paper, we focus on local convergence of target formations with general shapes, by assuming that initial formations are minimally and infinitesimally rigid, and are close to a target formation, which is a common assumption that has been widely used in many papers on rigid formation control; see e.g. \cite{krick2009stabilisation, oh2015survey, Mou_Problematic, Hector2016maneuvering, sun2016exponential}. We note that local stability is also a practical problem, arising when wind disturbs a formation away from its desired shape, and the original shape has to be recovered. Local convergence of formation shape also has practical significance. For example, agents can firstly assemble an approximate formation close to the desired shape, and then apply the control law to achieve the target formation guaranteed by the local convergence.
Global analysis of formation convergence and stability  is however only available for some particular and simple formation shapes (see e.g., \cite{dorfler2010geometric} for 2-D triangular shape,  \cite{park2014stability} for 3-D tetrahedral shape, and \cite{chen2017global} for 2-D triangulated formations), while global analysis for rigid formation systems with general shapes is a challenging and open problem. Global analysis of formation convergence is therefore beyond the scope of this paper and will not be discussed here.

A preliminary version of this paper was presented in \cite{sun2016quantization}.
The extensions of this paper
compared to \cite{sun2016quantization} include  {\color{black} detailed proofs for all the key results which were omitted in \cite{sun2016quantization}, examples on several quantizer functions, and  a new section to discuss the formation convergence when using an asymmetric uniform quantizer. For formation control with binary distance measurement, we also present several new properties for such formation control systems, including the independence of a global coordinate frame, boundedness of the control input, and a subspace-preserving principle which will help analyze the transient behavior and implementation issue of such control laws. Furthermore, several simulation results which support the theoretical analysis are provided in this extended paper. }

The remaining parts are organized as follows.
Section \ref{sec:background} briefly reviews some background on graph rigidity and two commonly-used quantizer functions. Section \ref{sec:convergence_quantizer} discusses the convergence of the formation systems under two quantized formation controllers.  In Section \ref{sec:binary_quantizer} we show a special quantized formation controller with binary distance information. Section \ref{sec:asymmetric_quantizer} focuses on the case of an asymmetric uniform quantizer and its performance. Some illustrative examples are provided in Section \ref{sec:quantization_simulations}. Section \ref{sec:Conclusions} concludes this paper.

\section{Background and preliminaries} \label{sec:background}
\subsection{Notations}  \label{subsec:notations}
Most notations used in this paper are fairly standard, and here we introduce some special notations that will find use  in later analysis. {\color{black} The  operator } $\text{col}(\cdot)$ defines the stacked column vector. For a given matrix $A\in\mathbb{R}^{n\times m}$, define $\overline A := A \otimes I_d \in\mathbb{R}^{nd\times md}$, where the symbol $\otimes$ denotes the Kronecker product and $I_d$ is the $d$-dimensional identity matrix with $d=\{2, 3\}$. We denote by $||x||$ the Euclidean norm of a vector $x$, by $\hat x := \frac{x}{||x||}$ the unit vector of $x \neq 0$, and by $\tilde x := \frac{1}{||x||}$ the reciprocal of the norm of $x \neq 0$. For a stacked vector  $x := \begin{bmatrix}x_1^{\top}, x_2^{\top}, \dots,  x_k^{\top}\end{bmatrix}^{\top}$ with $x_i\in\mathbb{R}^{l}, i\in\{1,\dots,k\}$, we define the block diagonal matrix $D_x := \operatorname{diag}\{x_i\}_{i\in\{1,\dots,k\}} \in\mathbb{R}^{kl\times k}$.
\subsection{Graph  rigidity}  \label{subsec:rigidity}

{\color{black} Now we review some basic graph theoretic tools for modelling  formation control systems. For more background on graph theory in multi-agent systems and formation control, we refer the reader to \cite{mesbahi2010graph}. }
Consider an undirected graph with $m$ edges and $n$ vertices, denoted by $\mathcal{G} =( \mathcal{V}, \mathcal{E})$  with vertex set $\mathcal{V} = \{1,2,\cdots, n\}$ and edge set $\mathcal{E} \subseteq \mathcal{V} \times \mathcal{V}$.  The neighbor set $\mathcal{N}_i$ of node $i$ is defined as $\mathcal{N}_i: = \{j \in \mathcal{V}: (i,j) \in \mathcal{E}\}$.
 We define an oriented incidence matrix $B\in\mathbb{R}^{n\times |\mathcal{E}|}$ for the undirected graph $\mathcal{G}$ by assigning an \emph{arbitrary} orientation for each edge. Note that for a rigid formation modelled by an \emph{undirected} graph considered in this paper, the orientation of each edge for writing the incidence matrix can be defined arbitrarily and the stability analysis in the next sections remains unchanged.   By doing this, we define the entries of $B$ as $b_{ik} = -1$, if $i = {\mathcal{E}_k^{\text{tail}}}$, or $b_{ik} = +1$, if $i = {\mathcal{E}_k^{\text{head}}}$, or $b_{ik} = 0$ otherwise,
where $\mathcal{E}_k^{\text{tail}}$ and $\mathcal{E}_k^{\text{head}}$ denote the tail and head nodes, respectively, of the edge $\mathcal{E}_k$, i.e. $\mathcal{E}_k = (\mathcal{E}_k^{\text{tail}},\mathcal{E}_k^{\text{head}})$. For a connected undirected graph, one has $\text{null}(B^{\top}) = \text{span}\{\mathbf{1}_n\}$. 

We denote by $p=[p_1^{\top}, \, p_2^{\top}, \cdots, \, p_n^{\top}]^{\top} \in \mathbb{R}^{dn}$ the stacked vector of all the agents' positions $p_i\in \mathbb{R}^d$. We also define \emph{non-collocated} positions for all agents as those positions for which  $p_i \neq p_j$ for all $(i,j) \in \mathcal{E}$.  The pair $(\mathcal{G}, p)$ is said to be a framework of $\mathcal{G}$ in $\mathbb{R}^d$. The incidence matrix $B$ defines the sensing topology of the formation, i.e. it encodes the set of available relative positions that can be measured by the agents. One can construct the stacked vector of available relative positions by
\begin{equation}
z = {\overline B}^{\top} p,  \label{z_equation}
\end{equation}
where each element $z_k  \in \mathbb{R}^d$ in $z$ is the relative position vector for the vertex pair defined by the edge $\mathcal{E}_k$.

Let us now briefly recall the notions of infinitesimally rigid framework and minimally rigid framework from \cite{hendrickson1992conditions} and \cite{anderson2008rigid}.
Let us define the edge function  $f_{\mathcal{G}}(p) := \mathop{\text{col}}\limits_{k}\big(\|z_k\|^2\big)$.  We denote the Jacobian of $\frac{1}{2}f_{\mathcal{G}}(p)$ by $R(z)$, which is called the {\it rigidity matrix}. An easy calculation shows that $R(z) = D^{\top}_z\overline B^{\top}$. A framework $(\mathcal{G}, p)$ is {\it infinitesimally rigid} if $\text{rank} (R(z)) = 2n-3$ when it is embedded in $\mathbb{R}^2$ or if $\text{rank} (R(z)) = 3n-6$ when it is embedded in $\mathbb{R}^3$. If additionally  $|\mathcal E|=2n-3$ in the 2D case or $|\mathcal E|=3n-6$ in the 3D case then the framework is called {\it minimally rigid}.    In this paper we assume that the target formation is  infinitesimally and minimally rigid, while the convergence results obtained in this paper can be extended to   non-minimally but still infinitesimally rigid target formations   by following the analysis in \cite{Mou_Problematic}  or \cite{sun2016exponential}.

\subsection{Quantizer functions} \label{sec:quantizer_definition}
In this paper, we mainly consider two types of quantizers: the uniform quantizer and the logarithmic quantizer \cite{cai2011quantized, ceragioli2011discontinuities, frasca2012continuous,  liu2012quantization, guo2013consensus}. In later sections we will also consider two special quantizers, namely a quantizer involving the signum function and the asymmetric uniform quantizer, derived from these quantizer functions.

\subsubsection{Definition of the quantizers}
The \emph{symmetric} uniform quantizer is a map $q_u: \mathbb{R} \rightarrow \mathbb{R}$ such that
\begin{align} \label{eq:sym_uniform_quantizer}
q_u(x) = \delta_u \left( \left \lfloor\frac{x}{\delta_u} \right \rceil \right ),
\end{align}
where $\delta_u$ is a positive number and $\lfloor{a}\rceil, a \in \mathbb{R}$ denotes the nearest integer  to $a$.  We also define $\lfloor{\frac{1}{2}+h}\rceil = h$ for any $h \in \mathbb{Z}$.

The \emph{logarithmic} quantizer is an odd map $q_l: \mathbb{R} \rightarrow \mathbb{R}$ such that
     \begin{equation} \label{eq:logarithmic_quantizer}
     q_l(x) =  \left\{
       \begin{array}{cc}
       \text{exp}({q_u(\text{ln}x)})  \;\;&\text{when} \,\, x>0;  \\
       0                    \;\;& \text{when} \,\, x=0;  \\
       -\text{exp}({q_u(\text{ln}(-x))}) \;\;&\text{when} \,\, x<0.
       \end{array}
      \right.
      \end{equation}
where $\text{exp}(\cdot)$ denotes the exponential function.

\subsubsection{Properties of the quantizers}
For the uniform quantizer, the quantization error is always bounded by $\delta_u/2$, namely $|q_u(x) - x| \leq \frac{\delta_u}{2}$ for all $x \in \mathbb{R}$.

For the logarithmic quantizer, it holds that $q_l(x)x \geq 0$,  for all $x \in \mathbb{R}$,
and the equality holds if and only if $x = 0$.  The quantization error for the logarithmic quantizer is bounded by $|q_l(x) - x| \leq \delta_l |x|$, where the parameter $\delta_l$
is determined by $\delta_l = \text{exp}({\frac{\delta_u}{2}})-1$ (note that $\delta_l > 0$ because $\delta_u >0$).

The above definitions for scalar-valued uniform and logarithmic quantizers can be generalized to  vector-valued quantizers for a vector in a component-wise manner.
For an illustration of a logarithmic quantizer function, see Figure \ref{fig:logarithmic_quantizer}. Note that in Section \ref{sec:asymmetric_quantizer} we will further consider an \emph{asymmetric} uniform quantizer, and will provide some comparisons between a \emph{symmetric} uniform quantizer and an \emph{asymmetric} uniform quantizer.

\begin{figure}
  \centering
  \includegraphics[width=70mm]{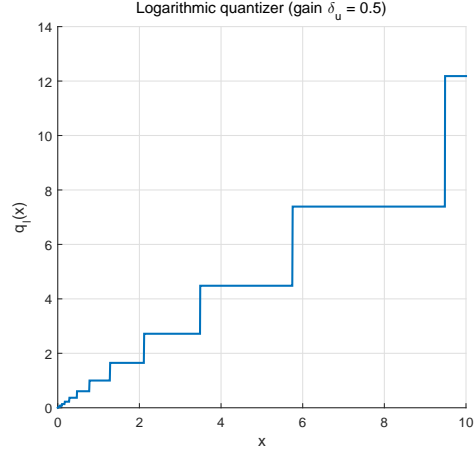}
  \caption{Logarithmic quantizer function with the gain $\delta_u = 0.5$, defined in \eqref{eq:logarithmic_quantizer}. }
  \label{fig:logarithmic_quantizer}
\end{figure}

\subsection{Nonsmooth analysis}
Consider a differential  equation
\begin{align} \label{eq:discontinuous}
\dot x(t) = X(x(t)),
\end{align}
where $X: \mathbb{R}^d \rightarrow \mathbb{R}^d$  is a vector field which is measurable but discontinuous. The existence of a continuously differentiable solution to \eqref{eq:discontinuous} cannot be guaranteed due to the discontinuity of $X(x(t))$. Also, as shown in \cite{ceragioli2011discontinuities}, the   Caratheodory solutions (for definitions, see \cite{cortes2008discontinuous}) may not exist
from a set of initial conditions of measure zero in quantized control systems. Therefore, we understand the solutions to the quantized rigid formation system in the sense of Filippov \cite{filippov1988differential}. We first introduce the Filippov set-valued map.

\begin{definition}
Let $\mathcal{D}(\mathbb{R}^d)$ denote the collection of all subsets of $\mathbb{R}^d$.
The Filippov set-valued map $F[X]: \mathbb{R}^d \rightarrow \mathcal{D}(\mathbb{R}^d)$ is defined by
\begin{align}
\mathcal{F}[X](x) \triangleq \bigcap \limits_{\delta > 0} \bigcap \limits_{\mu ({\cal S})=0} \overline{{\rm co}}\{X(\mathbb{B}(x, \delta)\backslash {\cal S}) \}, \quad x\in \mathbb{R}^d
\end{align}
where $\overline{{\rm co}}$ denotes convex closure,  $\cal S$ is the set of $x$ at which $X(x)$ is discontinuous, $\mathbb{B}(x, \delta)$ is the open ball of radius $\delta$ centered at $x$,   and $\bigcap_{\mu ({\cal S})=0}$ denotes the intersection over all sets $\cal S$ of Lebesgue measure zero.
\end{definition}
  Because of the way the Filippov set-valued map is defined, the value of $\mathcal{F}[X]$ at a discontinuous point $x$ is independent of the value of the vector field $X$ at $x$. Filippov solutions are absolutely continuous  curves that satisfy almost everywhere the differential inclusion $\dot x(t) \in \mathcal{F}[X](x)$ defined above. {\color{black} To keep the paper concise, we do not review further background knowledge on nonsmooth analysis and Filippov solutions. More properties of the Filippov solution and examples of how to compute a Filippov set-valued map can be found in the tutorial paper \cite{cortes2008discontinuous}.  }
\section{Formation control with quantized measurements} \label{sec:convergence_quantizer}
\subsection{Quantized formation controllers}
In rigid formation control, {\color{black} the target formation shape is described by a certain set of distances and all agents aim to stabilize cooperatively their distances to the desired ones. }  If one assumes perfect measurements, a commonly-used formation controller can be written as (see e.g. \cite{anderson2007control, Hector2016maneuvering})
\begin{equation} \label{eq:original_position_system}
\dot{p}_i = -\sum_{k=1}^{|\mathcal{E}|} b_{ik}  (||z_k|| - d_k)\, \hat z_k,
\end{equation}
 where  $d_k$ is the desired distance for edge $k$ which is adjacent with agent $i$, $z_k : = p_j - p_i$ is the relative position vector for edge $k$ that associates agents $i$ and $j$, $\hat z_k : = \frac{z_k}{\|z_k\|}$ is the unit vector denoting the bearing information for edge $k$, and other symbols  appeared in \eqref{eq:original_position_system} have been introduced in Sections \ref{subsec:notations} and \ref{subsec:rigidity}.  
{\color{black}
It is clear from \eqref{eq:original_position_system} that each agent needs to obtain relative position information to its neighbors. More precisely, agents in many, if not most, practical situations do not measure relative positions $z_k$ directly but the distance $||z_k||$ and the direction $\hat z_k$ as independent terms for the control action \eqref{eq:original_position_system}. For example, in ground and aerial robotic applications, inter-agent distances and bearings are measured by different kinds of sensors, e.g., measuring the round trip time of a radio signal for observing distances, and vision cameras or directional antennas for observing bearings. In particular, as we will show in Lemma  \ref{lemma:coordinate_frame}, the relative direction $\hat z_k$ can be given with respect to an agent's local coordinate frame.

Since the inter-agent distances and bearings are so often observed by different sensors, it is reasonable to quantize independently the variables $||z_k||$ and $\hat z_k$ in \eqref{eq:original_position_system}. Consequently, 
 in the presence of quantized sensing, the right-hand side of the control action \eqref{eq:original_position_system} needs modification.
Here we assume that the distance measurement is quantized, and the bearing measurement is unquantized. This assumption is not arbitrary but based on recent results in distance-based formation control. The work in \cite{bishop2015distributed} shows that, in distance-based formation control, accuracy in bearing measurements are not critical. In particular, deviations of up to but less than $\pm$ 90 degrees with respect to the optimal direction, i.e., the one linking two (or more) neighboring agents, do not compromise the stability and convergence of the formation. On the other hand, a distance-based formation is very sensitive to measurement errors, biases or inconsistencies of  an inter-agent distance by its two adjacent agents as  has been rigorously shown in \cite{Mou_Problematic,de2015controlling}. Therefore, we focus our analysis on studying the quantization in sensors that measure the distances and not the bearings, 
i.e., in this paper we consider the control term with quantized distance measurements in the form as $q(||z_k|| - d_k)\, \hat z_k$. In practice, this choice of nonquantization of the bearing is also reasonable because the bearing measurement is always bounded (described by a unit vector, or by an angle in $[-\pi, \pi)$ in the 2-D case). A normal digital sensor, say a 10-bit uniform quantizer, applying to bearing measurements gives rather accurate measurement with very small error to the true bearing, and as discussed in \cite{bishop2015distributed}, this is not an important issue regarding the stability of the formation. This is not the case for distance measurements which may have larger magnitudes.}

\begin{remark}
One may wonder why there is not use of  the quantization feedback in the form of $q(||z_k||)$, i.e., the direct quantized distance measurement,  in the control \eqref{eq:position_system} below.   We note three reasons for choosing $q(||z_k|| - d_k)$ instead of $q(||z_k||)$, {\color{black}from the viewpoint of using quantization as a necessity (arising from limited measurement capabilities), and the advantages that are brought about by adopting such a quantization strategy:}
\begin{itemize}
\item In rigid formation control, the control objective is to stabilize the actual distances between neighboring agents to prescribed values. If one chooses the quantization strategy in the form of $q(||z_k||)$, then this control objective may not be achieved. To this end, the quantization strategy $q(||z_k|| - d_k)$ used in \eqref{eq:position_system} can be interpreted  as   arising from a digital distance sensor with an embedded or prescribed offset (where the offset is the desired distance $d_k$), which is practical in real-world applications.
\item In the case of non-uniform quantizers (e.g., logarithmic quantizer), the quantization accuracy (or resolution) usually increases when the quantizer input approaches closer to the desired state (which is the origin in this case). Thus, when the formation approaches closer to the target formation, a higher quantization accuracy (if possible) is required, and this  cannot be achieved if one uses the quantization function (e.g., logarithmic quantization) on the actual distance in the form of $q(||z_k||)$.   Furthermore, such a quantization is  appealing as a design choice as one can have finite time convergence (as proved in the later section), which is another advantage from the formation convergence viewpoint.
\item  We will further show in Section \ref{sec:binary_quantizer} that  the chosen quantization strategy $q(||z_k|| - d_k)$ will specialize  to a simple and effective quantizer with \emph{coarse binary} distance measurement.   This also brings about the possibility of finite time formation convergence, as will be discussed in Section \ref{sec:binary_quantizer}.
\end{itemize}
\end{remark}

{\color{black}With the above considerations, we rewrite the control action \eqref{eq:original_position_system} as }
\begin{equation} \label{eq:position_system}
\dot{p}_i = -\sum_{k=1}^{|\mathcal{E}|} b_{ik} \, q(||z_k|| - d_k)\, \hat z_k,
\end{equation}
where $q$ is a quantization function, which can be the uniform quantizer or the logarithmic quantizer defined in Section \ref{sec:quantizer_definition}. We also assume that all the agents use the same quantizer $q(\cdot)$, and their initial positions start with non-collocated positions (which ensures $z_k(0) \neq 0$ for all $k$).

In the presence of quantized measurement and feedback, the right-hand side of  \eqref{eq:position_system} is discontinuous and we will consider the following differential inclusion
\begin{align}\label{eq:differential_inclusion_quantization}
\dot{p}_i  &  \in \mathcal{F} \left[ \sum_{k=1}^{|\mathcal{E}|} b_{ik} \, q(||z_k|| - d_k)\, \hat z_k  \right].
\end{align}

  In the following, we define the distance error for edge $k$ as $e_k = ||z_k|| - d_k$.    We then calculate the differential inclusion $\mathcal{F}(q(e_{k}))$ which will be used in later analysis. In the case of a symmetric uniform   quantizer, the  differential inclusion $\mathcal{F}(q_u(e_{k}))$ can be calculated as
     \begin{align}
     \mathcal{F}(q_u(e_{k})) =  \left\{
       \begin{array}{cc}
       h \delta_u,  \;\;&  e_{k} \in \left( (h-\frac{1}{2})\delta_u,  (h + \frac{1}{2})\delta_u  \right), h \in \mathbb{Z};  \\ \nonumber
       [h \delta_u, (h+1) \delta_u],  \;\;&  e_{k} = (h + \frac{1}{2})\delta_u  , h \in \mathbb{Z}.
       \end{array}
      \right.
      \end{align}
Note that $e_{k} \mathcal{F}(q_u(e_{k})) \geq 0$ for all $e_{k}$, and $e_{k} \mathcal{F}(q_u(e_{k})) = 0$ if and only if $e_{k} \in [-\frac{\delta_u}{2}, \frac{\delta_u}{2}]$.  We refer the reader to Section 2.3 for the definition and meaning of notations such as $\delta_u$ and $h$.

In the case of a logarithmic quantizer, the  differential inclusion $\mathcal{F}(q_l(e_{k}))$  can be calculated as
     \begin{align}
     \mathcal{F}(q_l(e_{k})) =  \left\{
       \begin{array}{cc}
       \text{sign}(e_{k}) \text{exp}({q_u(\text{ln}|e_{k}|)}),     \;\;&  e_{k} \neq e^{(h + \frac{1}{2})\delta_u}  , h \in \mathbb{Z};  \\ \nonumber
       [\text{exp}({h \delta_u}), \text{exp}({(h+1) \delta_u})],  \;\;&  e_{k} = e^{(h + \frac{1}{2})\delta_u}, h \in \mathbb{Z}.
       \end{array}
      \right.
      \end{align}
Also note that $e_{k} \mathcal{F}(q_l(e_{k})) > 0$ for all $e_{k} \neq 0$, and   $e_{k} \mathcal{F}(q_l(e_{k})) = 0$ if and only if $e_{k}  = 0$.

We define the distance error vector as $e = [e_1, e_2, \dots, e_m]^{\top}$. Then in a compact form, one can rewrite   the dynamics of \eqref{eq:differential_inclusion_quantization} as
\begin{align} \label{eq:compact_position_system}
\dot p &\in \mathcal{F}\Big[- \overline BD_{\hat z} \, q\Big(e(\mathop{\text{col}}\limits_{k}\big(\|z_k\|\big)\Big)\Big].
\end{align}
In order not to overload the notation,   here by $\hat z$ we exclusively mean the vector-wise normalization of $z$, therefore $D_{\hat z}$ in the above equation and in the sequel is defined as $D_{\hat z} = \operatorname{diag}\{\hat z_1, ..., \hat z_m\}$. This notation rule will also be applied to $\tilde z$ in the sequel.
Note that the differential inclusion $\mathcal{F}(q(e))$ with the vector $e$ is defined according to the product rule of  Filippov's calculus properties (see \cite{paden1987calculus}).

\begin{example} \label{exp:3D_formation}
We show an example to illustrate the formation control system with quantized measurements. Consider a formation system aiming to achieve a double tetrahedron shape in 3-D space (see Figure \ref{fig:tetrahedron_shape}), which consists
of five agents labeled by $1, 2, 3, 4, 5$ and nine edges. This formation is minimally rigid. For the purpose of writing an oriented incidence matrix, suppose that the nine edges are oriented from $i$
to $j$ just when $i < j$.  Then we can number the edges in the following order: $12, 13, 14, 23, 34, 24, 25, 35, 45$. Thus, the following oriented
incidence matrix for the undirected graph in Figure \ref{fig:tetrahedron_shape}  can be obtained

\begin{figure}
  \centering
  \includegraphics[width=50mm]{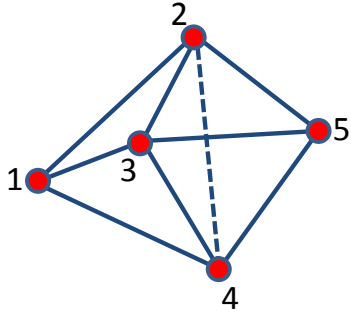}
  \caption{ A 3-D double tetrahedron formation shape, with 5 agents and 9 distances.}
  \label{fig:tetrahedron_shape}
\end{figure}

\begin{equation}
B = \left[
\begin{array}{ccccccccc}
-1 & -1 & -1 & 0 & 0 & 0 & 0 & 0 & 0\\
1 & 0 & 0 & -1 & 0 & -1 & -1 & 0 & 0\\
0 & 1 & 0 & 1 & -1 & 0 & 0 & -1 & 0\\
0 & 0 & 1 & 0 & 1 & 1 & 0 & 0 & -1\\
0 & 0 & 0 & 0 & 0 & 0 & 1 & 1 & 1
\end{array}. \label{incidence_matrix}
\right]
\end{equation}
The composite relative position vector $z$ is then defined according to the orientation for each edge, as in the incidence matrix $B$, in the form of $z = {\overline B}^{\top} p$. As an
example, one has $z_1 =p_2-p_1$, i.e. the vector $z_1$ at edge 1 is defined by the relative position between agent 2 and agent 1. The formation dynamics for agent 1 can be written as 
\begin{align}\label{eq:differential_inclusion_quantization_1}
\dot{p}_1  &  \in \mathcal{F} \left[  q(||z_1|| - d_1)\, \hat z_1 + q(||z_2|| - d_2) \hat z_2 + q(||z_3|| - d_3) \hat z_3 \right].
\end{align}
and similarly one can obtain the system equation for other agents.

By defining the matrix $\overline B$ and $D_{\hat z}$, one can obtain compact equations of system dynamics
in the compact form of \eqref{eq:compact_position_system}. 
\qedsymbol
\end{example}

\subsection{Properties of quantized formation control systems} \label{sec:properties_system}

We first discuss the solution issue of the formation control system \eqref{eq:differential_inclusion_quantization}. However, it is more convenient to focus on the dynamics of the relative position vector $z$, which can be derived from \eqref{eq:differential_inclusion_quantization}  as follows
\begin{align} \label{eq:z_equation_solution}
\dot z = \overline B^{\top} \dot p
\in   \mathcal{F}\Big[- \overline B^{\top} \overline BD_{\hat z} \, q\Big(e(\mathop{\text{col}}\limits_{k}\big(\|z_k\|\big)\Big)\Big].
\end{align}
First note that at any non-collocated finite initial point $p(0)$, the right-hand side of \eqref{eq:differential_inclusion_quantization} and of \eqref{eq:z_equation_solution} is  measurable and locally essentially
bounded.
Thus, the existence of a local Filippov solution of \eqref{eq:differential_inclusion_quantization} and of \eqref{eq:z_equation_solution} starting at such initial points is guaranteed.

We then derive a dynamical system from \eqref{eq:compact_position_system} to describe the evolution of the distance error vector $e$.
According to the definition of the distance error $e_k$, $e_k$ is a smooth function of $z_k$. Thus, by using the calculus property (see \cite{cortes2008discontinuous})  and the set-valued Lie derivative computation theorem (see \cite{bacciotti1999stability}), one can show  $\dot e_k$ exists and 
 $\dot e_k = \frac{1}{\|z_k\|} z_k^{\top} \dot z_k$   holds almost everywhere.  The dynamics for the distance error vector $e$ can be obtained in a compact form as
\begin{align} \label{error system}
\dot e &=   D_{\tilde z} D_{z}^\top \dot z = D_{\tilde z} D_{z}^\top \overline B^{\top} \dot p  = D_{\tilde z} R(z) \dot p,\,\,\,\,  \text{a. e.} \nonumber \\
& \in  - \, \mathcal{F} \left[ D_{\tilde z} R(z) R^{\top}(z) D_{\tilde z} q(e) \right],\,\,\,\,\text{a. e.}
 \end{align}
 A more general compact form of the system equation $\dot e$ can be found in \cite[Section III]{Hector2016maneuvering}.
Again, the existence of a local Filippov solution of \eqref{error system} starting with a  non-collocated finite initial point $p(0)$ is guaranteed.  In the next section, we will also show that the solutions to \eqref{error system} (as well as the solutions  to \eqref{eq:differential_inclusion_quantization} and   \eqref{eq:z_equation_solution}) are \emph{bounded} and can be extended to $t \rightarrow \infty$  when agents' initial positions are chosen non-collocated and close to a target formation shape.  Also, as shown in \cite{Mou_Problematic}, when the formation shape is close to
the desired one, the entries of the matrix $R(z) R^{\top}(z)$ are
continuously differentiable functions of $e$. Since the nonzero
entries of the diagonal matrix $D_{\tilde z}$ are of the form $\frac{1}{\|z_k\|}$
which are also continuously differentiable functions of $e$, we
conclude that the system described in \eqref{error system} is a self-contained
system, and we will call it \emph{the distance error} system in the
sequel.

\begin{example} (Continued)
We again use the formation shape shown in Figure \ref{fig:tetrahedron_shape} to illustrate the derivation of the above system equations. 

According to the definition of the relative position vector $z$ one can derive the compact form of the system dynamic equation for $z$, as shown in \eqref{eq:z_equation_solution}. 
From the construction of the incidence matrix $B$ and the relative position vector $z$, the rigidity matrix can be obtained as
\begin{equation}
R = \left[
\begin{array}{ccccc}
-z_1^{\top} & z_1^{\top} & 0 & 0 & 0\\
-z_2^{\top} & 0 & z_2^{\top} & 0 & 0\\
-z_3^{\top} & 0 & 0 & z_3^{\top} & 0\\
0 & -z_4^{\top} & z_4^{\top} & 0 & 0\\
0 & 0 & -z_5^{\top} & z_5^{\top} & 0\\
0 & -z_6^{\top} & 0 & z_6^{\top} & 0\\
0 & -z_7^{\top} & 0 & 0 & z_7^{\top}\\
0 & 0 & -z_8^{\top} & 0 & z_8^{\top}\\
0 & 0 & 0 & -z_9^{\top} & z_9^{\top}
\end{array} \label{rigidity_matrix}
\right]
\end{equation}

From the expression of the matrix $R(z)$ in \eqref{rigidity_matrix}, it is obvious that the entries of the matrix product $R(z)R^{\top}(z)$ are either zero, or  inner products of relative position vectors in the form of $z_i^{\top} z_j$, which are functions of the distance error vector $e$ (for detailed analysis, see e.g., \cite{Mou_Problematic, sun2016exponential}). Since the   diagonal matrix $D_{\tilde z}$
is defined as $D_{\tilde z} = \operatorname{diag}\{\tilde z_1, ..., \tilde z_9\}$ with $\tilde z_k = \frac{1}{\|z_k\|}$, it is clear that the entries of the matrix $D_{\tilde z}$ are also functions of $e$, and therefore the entries in the matrix product $D_{\tilde z} R(z) R^{\top}(z) D_{\tilde z}$ are functions of $e$. Hence, the distance error system in \eqref{error system} is a self-contained system, for which we can apply the Lyapunov argument to show its stability.  
Note the compact form of the error system \eqref{error system} can be derived by the definition of the distance error $e$.   \qedsymbol
\end{example}

Finally, we show some additional properties of the formation control system with quantized information. Note that through this paper we assume that the underlying graph modelling inter-agent interactions is undirected.

\begin{lemma}
In the presence of the uniform/logarithmic quantizer, the formation centroid remains stationary.
\end{lemma}

\begin{proof}
Denote by $p_c \in \mathbb{R}^d$ the center of the mass of the formation, i.e., $p_c = \frac{1}{n} \sum_{i=1}^n p_i = \frac{1}{n}(\mathbf{1}_n \otimes I_{d \times d})^{\top} p$.
By applying  the calculus property for the set-valued Lie derivative (see \cite{cortes2008discontinuous} or \cite{bacciotti1999stability}), one has
\begin{align}
\dot p_c(t) = & \frac{1}{n}(\mathbf{1}_n \otimes I_{d \times d})^{\top} \dot p  \nonumber \\
\in & -\frac{1}{n}(\mathbf{1}_n \otimes I_{d \times d})^{\top} R^{\top}(z) D_{\tilde z}  \mathcal{F}[q(e(z))]\,\,\, \text{for a.e. } t.
\end{align}
Note that $(\mathbf{1}_n \otimes I_{d \times d})^{\top} R^{\top}(z) = \mathbf{0}$. Therefore,
\begin{align}
\dot p_c(t)  \in & -\frac{1}{n}(\mathbf{1}_n \otimes I_{d \times d})^{\top} R^{\top}(z) D_{\tilde z}  \mathcal{F}[q(e(z))] = \{0\}\, \text{for a.e. } t.
\end{align}
Thus $\dot p_c = 0$ for a.e.  $t$, which indicates that the position of the formation centroid remains constant.
\end{proof}

\begin{lemma} \label{lemma:coordinate_frame}
To implement the control, each   agent can use its own local coordinate system to measure the   relative position (quantized distance and unquantized bearing) of its neighbors, and a global coordinate system is not involved.
\end{lemma}

\begin{proof}
The key part in the proof of local coordinate requirement is to show that the control function for all the agents is an \textit{SE(N)}-invariant function {\footnote{A function $f$ is said to be \textit{SE(N)}-invariant if for
all $R \in  SO(N)$ and all $x_1, \ldots, x_n, \omega \in  \mathbb{R}^N$, there holds 
$Rf(x_1,\ldots, x_n) = f(Rx_1 + \omega,\ldots, Rx_n +  \omega)$. }}.  
The control function for agent $i$ is $f_i = \sum_{k=1}^{|\mathcal{E}|} b_{ik} \, q(||z_k|| - d_k)\, \hat z_k$. Given an arbitrary coordination rotation $R \in  SO(N)$ and  displacement of origin  $\omega \in  \mathbb{R}^N$, {there holds \color{black} $\|Rp_i + \omega - (R p_j +\omega)\| = \|Rp_i - Rp_j\| = \|p_i - p_j\|$ and therefore $q(||z_k|| - d_k)$ is invariant under any action of rotation $R$ and translation $\omega$.   }
Thus, one has
\begin{align}
f_i(Rp_1 + \omega,\ldots, Rp_n +  \omega) &= \sum_{k=1}^{|\mathcal{E}|} b_{ik} \, q(||z_k|| - d_k)\, R \, \hat z_k \nonumber \\
&= R \sum_{k=1}^{|\mathcal{E}|} b_{ik} \, q(||z_k|| - d_k)\,  \hat z_k \nonumber \\
&= R f_i (p_1, \ldots, p_n)
\end{align}
and the statement is proved. 
\end{proof}
 
Note that the above lemma implies  the $SE(N)$ invariance (i.e., translational and rotational invariance) \cite{vasile2017translational} of the proposed formation controller, which enables a convenient implementation of the quantized formation control law without coordinate frame alignment for all of the agents. We refer the readers to \cite{vasile2017translational} for a general treatment on coordinate frame issues in networked control systems.

\subsection{Convergence analysis}
In this section we aim to prove  the following convergence result.

\begin{theorem} \label{theorem:convergence_quantization}
Suppose the target formation is infinitesimally
and minimally rigid and the  formation controller with quantized measurement  is applied.
\begin{itemize}
\item In the case of a uniform quantizer, the formation converges locally to an approximately correct and static shape defined by the set $F_{\text{approx}} = \{e | e_{k} \in [-\frac{\delta_u}{2}, \frac{\delta_u}{2}], k \in \{1, \dots, |\mathcal{E}|\}\}$;
\item In the case of a logarithmic quantizer,   the formation  converges locally to a correct and static  formation shape.
\end{itemize}
\end{theorem}

In the proof we will use the  Lyapunov theory of nonsmooth analysis, for which we construct a Lyapunov function candidate as
\begin{align} \label{eq:lyapunov_quantization}
V(e) = \sum_{k=1}^{|\mathcal{E}|} V_k(e_k), \,\, \text{with}\,\,V_k(e_k) = \int_{0}^{e_k} q(s) \text{d} s.
\end{align}
Before giving the proof of Theorem \ref{theorem:convergence_quantization}, we first show some properties of the function $V$ defined in \eqref{eq:lyapunov_quantization}. For the definition of \emph{function regularity} in nonsmooth analysis, see e.g. \cite[Chapter 2]{clarke1998nonsmooth} or \cite[Page 57]{cortes2008discontinuous}.
\begin{lemma} \label{lemma:regularity}
The  function $V$ constructed in \eqref{eq:lyapunov_quantization} is positive semidefinite, and is regular everywhere.
\end{lemma}
\begin{proof}
The positive semidefiniteness of $V$ is obvious from the property of the quantization functions $q_u$ and $q_l$. Note that $V(e) = 0$ if and only if  $e \in \{e | e_{k} \in [-\frac{\delta_u}{2}, \frac{\delta_u}{2}], k \in \{1, \dots, |\mathcal{E}|\}\}$ for a uniform quantizer $q_u$, or when $e =0$ for a logarithmic quantizer $q_l$. The proof for the regularity is  omitted here but follows similarly to the proof of the previous paper \cite[Lemma 6]{liu2012quantization}. We note a key fact that supports the regularity statement of $V$: $V$ is  continuously differentiable almost everywhere, while at the nondifferentiable points $V$ has corners of \emph{convex} type. From the sufficient condition of regular functions stated in \cite[Page 200]{clarke2013functional}, \footnote{``Roughly speaking, we can think of regular functions as those that, at each point, are either smooth, or else have a corner of convex type'' \cite[Page 200]{clarke2013functional}.} this key fact implies that $V$ is regular everywhere.
\end{proof}

Furthermore, according to the definition of generalized derivative (see e.g.
\cite[Chapter 2]{clarke1998nonsmooth}), one can calculate the generalized derivative of $V_k$ (for the case of a uniform quantizer) as
\begin{align}
\partial V_k = \left\{
\begin{array}{cc}
       [h \delta_u, (h+1) \delta_u],  &   e_{k} = (h + \frac{1}{2})\delta_u  , h \in \mathbb{Z}    \\ \nonumber
       q(e_k),  &   \text{elsewise}   \\
       \end{array}
      \right.
\end{align}
Similarly, one can also calculate the generalized derivative of $V_k(e_k)$ for the case of a logarithmic quantizer (which is omitted here).  The  generalized derivative of $V(e)$ can be obtained by the product rule (see \cite[Page 50]{cortes2008discontinuous}).  Now we are ready to prove Theorem \ref{theorem:convergence_quantization}.

\begin{proof}
We choose the Lyapunov function constructed in \eqref{eq:lyapunov_quantization} for the distance error system \eqref{error system} with discontinuous right-hand side.
Note that $R(z)R^{\top}(z)$ and $D_{\tilde z}$ are positive definite matrices at the
desired formation shape.
Similarly to the analysis in   \cite{sun2016exponential} and in \cite{Hector2016maneuvering}, we  define a sub-level  set $\mathcal{B}(\rho)= \{e: V(e) \leq \rho\}$ for some suitably small  $\rho$,    such that when $e \in \mathcal{B}(\rho)$ the  formation is infinitesimally minimally rigid and the initial formation shape is close enough to the prescribed shape (which implies that  inter-agent collisions cannot be possible). Note that all these imply that $R(z)R^{\top}(z)$ and $D_{\tilde z}$ are \emph{positive definite} when $e \in \mathcal{B}(\rho)$. Note also that the defined sub-level  set $\mathcal{B}(\rho)$ is compact, and the matrix $Q(e) :=D_{\tilde z} R(z)R^{\top}(z) D_{\tilde z}$ is also \emph{positive definite} when $e \in \mathcal{B}(\rho)$.   As a consequence, in the following we rewrite the distance error system as
$\dot e \in \mathcal{F} \left[ -Q(e) q(e)\right]$.

The  regularity of $V$ shown in Lemma \ref{lemma:regularity} allows us to employ the nonsmooth Lyapunov theorem \cite[Section 2]{bacciotti1999stability} to develop the stability analysis.
We calculate the set-valued derivative of $V$ along the trajectory of the distance error system \eqref{error system}.  By applying the calculation rule for the set-valued derivative (see \cite[Pages 62-63]{cortes2008discontinuous}), one can obtain

\begin{align}
\dot V(e)_{\eqref{error system}} & \in \tilde{\mathcal{L}}_{\eqref{error system}}V(e)  = \{  a \in \mathbb{R}| \exists v \in \dot e_{\eqref{error system}}, \,\,\, \nonumber \\
 & \text{such that}\,\,\, \zeta^{\top} v= a, \forall \zeta \in \partial V(e) \}.
\end{align}
Note that the set $\tilde{\mathcal{L}}_{\eqref{error system}}V(e)$ could be empty, and in this case we adopt the convention that $\text{max} (\emptyset) = - \infty$.
When it is not empty, there exists   $v \in -Q(e) q(e)$ such that $\zeta^{\top} v= a$ for all $\zeta \in \partial V(e)$. A natural choice of $v$ is to set $v \in -Q(e) \zeta$, with which one can obtain $a = -q^{\top}(e)Q(e)q(e)$.

Let $\bar \lambda_{\text{min}}$ denote the smallest eigenvalue of $Q(e)$ when $e(p)$ is in the compact set $\mathcal{B} $ (i.e. $\bar \lambda_{\text{min}}  =  \mathop {\min }\limits_{e \in \mathcal{B}} \lambda  (Q(e)) >0$).   Note that $\bar \lambda_{\text{min}}$ exists because the set $\mathcal{B}$ is a compact set
and the eigenvalues of a matrix are continuous functions of the matrix elements, and $\bar \lambda_{\text{min}}>0$ because $Q(e)$ is \emph{positive definite} with $e \in \mathcal{B}(\rho)$ as mentioned above.  Then if the set   $\tilde{\mathcal{L}}_{\eqref{error system}}V(e)$ is not empty, one can show
\begin{align} \label{eq:Lya_derivative1}
\text{max}(\tilde{\mathcal{L}}_{\eqref{error system}}V(e))   \leq    - \bar \lambda_{\text{min}} q(e)^{\top}  q(e)
\end{align}
and if the set $\tilde{\mathcal{L}}_{\eqref{error system}}V(e)$ is empty, one has $\text{max}(\tilde{\mathcal{L}}_{\eqref{error system}}V(e)) = - \infty$. Note that both cases imply that $V$ is non-increasing, and consequently the Filippov solution $e(t)$ of \eqref{error system} is bounded. Thus, all solutions to \eqref{error system} (as well as the solutions to \eqref{eq:differential_inclusion_quantization}) are bounded and can be extended to $t = \infty$ (i.e., there is no finite escape time).

We now divide the rest of the proof into two parts, according to different quantizers:
\begin{itemize}
\item The case of uniform quantizers: it can be seen that $\text{max}(\tilde{\mathcal{L}}_{\eqref{error system}}V(e)) \leq 0$  for all $e \in \mathcal{B}(\rho)$ and
$0 \in \text{max}(\tilde{\mathcal{L}}_{\eqref{error system}}V(e))$ if and only if $e \in F_{\text{approx}}$. Also note that $F_{\text{approx}}$ is compact, and is positively invariant for the distance error system \eqref{error system} (i.e. if the initial formation is such that $e(0)  \in F_{\text{approx}}$, then all agents are static and $e(t)  \in F_{\text{approx}}$ for all $t$).  According to the nonsmooth invariance principle \cite[Theorem 3]{bacciotti1999stability}, the first part of the convergence result is proved. Since this is a convergence to a \emph{closed and bounded} set $F_{\text{approx}}$ (i.e., a compact set), and outside this set the set-valued derivative of $V$ along the trajectory of the distance
error system is always negative (i.e., $\text{max}(\tilde{\mathcal{L}}_{\eqref{error system}}V(e)) < 0$ for $e \in \mathcal{B}(\rho) \setminus F_{\text{approx}}$) while $\mathcal{B}(\rho)$ is also a compact set, the convergence to $F_{\text{approx}}$ is achieved within a \emph{finite time}.   Note also from \eqref{eq:position_system} the final formation is stationary because $\dot p(t) = 0$ for $e(t)  \in F_{\text{approx}}$.
\item The case of logarithmic  quantizers: it can be seen that $\text{max}(\tilde{\mathcal{L}}_{\eqref{error system}}V(e)) \leq 0$ for all $e \in \mathcal{B}(\rho)$ and
$0 \in \text{max}(\tilde{\mathcal{L}}_{\eqref{error system}}V(e))$ if and only if $e  = 0$. According to the nonsmooth invariance principle \cite[Theorem 3]{bacciotti1999stability}, the second part of the convergence result is proved. Also note from  \eqref{eq:position_system} the final formation is stationary.
\end{itemize}
The proof is thus completed.
\end{proof}

\begin{remark}(\textbf{Finite time convergence to a compact set})
In the above we have shown the trajectories of distance errors in the formation system under uniform quantization  converge to a bounded and closed set $F_{\text{approx}}$ within a finite time, the size of which also depends on the uniform quantizer errors. The key recipes to guarantee the finite time convergence are the following: (i) the set $F_{\text{approx}}$ and the sublevel sets $V(e)$ are compact sets; (ii) Outside the set $F_{\text{approx}}$  the set-valued derivative of $V$ along the trajectories of the distance error system is always negative; and (iii) the function $V(e)$ is a   strictly increasing function of $e$. An intuitive illustration of the finite time convergence of distance error trajectories  to the set $F_{\text{approx}}$ is shown in Figure \ref{figure:finite_time_set}.

\begin{figure} 
\centering
\includegraphics[width=75mm]{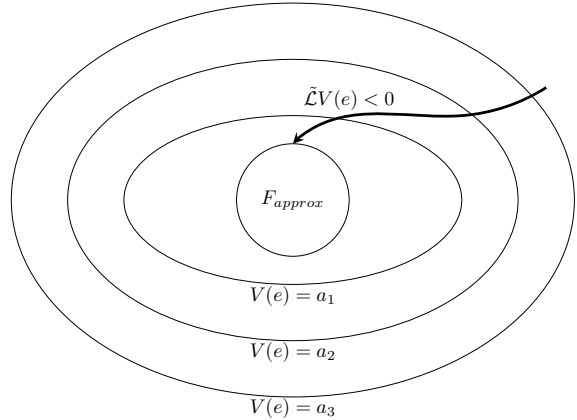}   
\caption{\small Illustration of finite time convergence to a compact set $F_{\text{approx}(e)}$ centered at $e = 0$. Outside this set $F_{\text{approx}}$  the set-valued derivative of $V$ along the trajectory of the distance
error system is always negative. In the figure, three level sets of $V(e)$ with $a_3 >a_2 >a_1$ are shown, which are compact sets with respect to $e$. Note as also shown in the figure $V(e)$ is a positive definite and strictly increasing function of $e$. }  \label{figure:finite_time_set}
\end{figure} 

\end{remark}

\begin{remark} \label{re:fixed_point}
We now show a stronger convergence result (i.e., convergence to a point in the set) in addition to the finite time convergence in the case of uniform quantizers \eqref{eq:sym_uniform_quantizer}. We observe that a sufficient condition for the position $p_i$ of agent $i$ to converge to a fixed point is that $\int_{0}^{\infty}\dot p_i(t)\text{d}t  < \infty$, which is true since (i) initially all agents are at finite positions (i.e., $p_i(0) < \infty$), and (ii) all $\dot p_i(t)$ (associated with the control input) are upper bounded and converge to the origin in \emph{finite time}. By the integration law this implies $p_i(t)_{t >T}$ is constant at a fixed position when $e(t)_{t >T} \in F_{\text{approx}}$ where $T$ is the finite settling time of convergence, which further implies that the distance error $e(t)$ converges to a fixed point in the set $F_{\text{approx}}$.
\end{remark}

\section{A special quantizer: formation control with   binary distance information} \label{sec:binary_quantizer}
\subsection{Rigid formation control with coarse distance measurements}

In this section we consider the special case in which each agent uses very coarse distance measurements, in the sense that   it only needs to detect whether the current distance to each of its neighbors is greater or smaller than the desired distance. This gives rise to a special quantizer defined by the following \emph{signum} function:
\begin{equation} \label{eq:definition_sign}
     \text{sign}(x) =  \left\{
       \begin{array}{cc}
       1  \;\;&\text{when} \,\, x>0;  \\ \nonumber
       0                    \;\;& \text{when} \,\, x=0;  \\ \nonumber
       -1 \;\;&\text{when} \,\, x<0.
       \end{array}
      \right.
      \end{equation}
Accordingly, we obtain the following   rigid formation control system with \emph{binary distance measurements}:
\begin{equation}\label{eq:position_sign2}
\dot{p}_i=  -\sum_{k=1}^{|\mathcal{E}|} b_{ik}  \text{sign}(||z_k|| - d_k)\, \hat z_k.
\end{equation}

\begin{remark}
Formation control using the signum function has been discussed in several previous papers. In \cite{zhao2014finite}, a finite-time convergence was established for stabilization of cyclic formations using binary bearing-only measurements.
A \textit{linear-consensus-based} formation control with coarsely quantized measurements was discussed in  \cite{Jafarian2015125}, while the implementation of the controller requires each agent to have knowledge of the global coordinate frame orientation. {\color{black} To be specific, the formation control law proposed in \cite{Jafarian2015125} is described by 
\begin{align}
u_i = \sum_{k=1}^{|\mathcal{E}|} b_{ik}\,\, \text{sign}(z_k - z_k^*),
\end{align}
where $z_k^* \in \mathbb{R}^d$ is the desired relative position for edge $k$, and the $\text{sign}(\cdot)$ function operates on each element of the $d$-dimensional vector $z_k - z_k^*$. In order to interpret the desired relative position vector $z_k^*$, all agents need to agree on a common orientation in their coordinate frames. In contrast, the proposed control law with coarse distance measurements in the form of \eqref{eq:position_sign2} does not require an orientation alignment of the agents' local coordinate frames. This is summarized in the following lemma. }
\end{remark}
{\color{black}
\begin{lemma}
To implement the control of \eqref{eq:position_sign2}, each  agent can use its own local coordinate system to measure the binary distances and bearings with respect to its neighbors, and a global coordinate system is not involved.
\end{lemma}
The proof of the above lemma is similar to the proof of Lemma~\ref{lemma:coordinate_frame} and is therefore omitted here. 
}

\begin{remark}
The  work closest to the controller setting in this section is the paper  \cite{liu2014controlling}, which studied  the stabilization control of a \textit{cyclic triangular} formation {\color{black} modelled in a directed graph} with the  controller \eqref{eq:position_sign2}. Here we extend such controllers to stabilize a general undirected formation   which is minimally and infinitesimally rigid.  The above controller \eqref{eq:position_sign2} can also be seen as a high-dimensional extension of the one-dimensional formation  controller studied in \cite{de2010control}.
\end{remark}

\begin{remark}
{\color{black}(\textbf{Boundedness of the control input \eqref{eq:position_sign2})} } Note that (since $b_{ij}\in\{-1,0,1\}$) the right-hand side of   \eqref{eq:position_sign2} is composed of the sum of a number of unit vectors multiplied by a signum function. This implies that the formation controller \eqref{eq:position_sign2} is of special interest in practice since the control action is explicitly upper bounded by the cardinality of the set of neighbors for each agent $i$, which prevents  potential implementation problems due to saturation. {\color{black} To be precise, the bound of the magnitude  of the control input for agent $i$ is derived by 
\begin{align}
\|u_i\| = \left\|\sum_{k=1}^{|\mathcal{E}|} b_{ik}  \text{sign}(||z_k|| - d_k)\, \hat z_k \right\| & \leq \sum_{k=1}^{|\mathcal{E}|} \left\| b_{ik} \text{sign}(||z_k|| - d_k)\, \hat z_k \right\| \nonumber \\
& \leq \sum_{k=1}^{|\mathcal{E}|} |b_{ik}| \left\| \text{sign}(||z_k|| - d_k) \right\|\, \left\|\hat z_k \right\|  \nonumber \\
& \leq \sum_{k=1}^{|\mathcal{E}|} |b_{ik}| \nonumber \\
& \leq |\mathcal{N}_i|.
\end{align}
}
\end{remark}

Again, we consider the Filippov solution to the formation control system \eqref{eq:position_sign2}. The differential inclusion $\mathcal{F}(\text{sign}(e_{k}))$  can be calculated as
     \begin{align}
     \mathcal{F}(\text{sign}(e_{k})) =  \left\{
       \begin{array}{cc}
       1     \;\;&  \|z_k\| > d_{k},   \\ \nonumber
       [-1, 1],  \;\;&  \|z_k\| = d_{k},   \\ \nonumber
       -1     \;\;&  \|z_k\| < d_{k}.
       \end{array}
      \right.
      \end{align}
In a compact form, the rigid formation system \eqref{eq:position_sign2} can be rewritten as
\begin{align} \label{eq:compact_position_sign2}
\dot p &\in  \mathcal{F}[- R^{\top}(z) D_{\tilde z} \text{sign}(e)],
\end{align}
where $\text{sign}(e)$ is defined in a component-wise way.

Note that the right-hand side of \eqref{eq:compact_position_sign2} is  measurable  and essentially bounded at any  non-collocated and finite point $p$, and the existence of a local Filippov solution to \eqref{eq:compact_position_sign2} is guaranteed from such an initial point $p(0)$. In the following analysis we will also show that the solutions are bounded and complete.

{\color{black}We now show a subspace-preserving property for the formation control system with the control law \eqref{eq:compact_position_sign2}.
\begin{lemma} \label{lemma_dimension_sign}
For the formation control system described by \eqref{eq:compact_position_sign2}, the (affine) subspace spanned by agents' solutions is invariant over time, i.e., the same as the (affine) subspace spanned by their initial positions. To be precise, there holds
\begin{align}  \label{eq:subspacepreserving_sense}
 \text{span}\left([p_1(t), p_2(t), \ldots, p_n(t)]\right)   = \text{span}\left([p_1(0), p_2(0), \ldots, p_n(0)]\right), \,\,\,\, \forall t \geq 0,
\end{align}
and
\begin{align}  \label{eq:affinesubspacepreserving_sense}
 \text{span}\left([z_1(t), z_2(t), \ldots, z_{|\mathcal{E}|}(t)]\right)   = \text{span}\left([z_1(0), z_2(0), \ldots, z_{|\mathcal{E}|}(0)]\right), \,\,\,\, \forall t \geq 0.
\end{align}

\end{lemma}
\begin{proof}
In \cite{sun2017dimensional}, it has been proved that any networked dynamical system with scalar couplings in the following form
\begin{align}  \label{eq:coupled_linear_wij}
\dot x_i(t) = \sum_{j =1}^{n} w_{ij}(t) x_j(t),
\end{align}
where $x_i \in \mathbb{R}^d$ and $w_{ij}$ is a scalar (constant or time-varying) coupling weight between agents $j$ and $i$, possesses a subspace-preserving property, in the sense that $\text{span}\left([x_1(t), x_2(t), \ldots, x_n(t)]\right)   = \text{span}\left([x_1(0), x_2(0), \ldots, x_n(0)]\right), \,\,\,\, \forall t \geq 0$. (See Theorem 1 and Corollary 1 of \cite{sun2017dimensional}). Note that the formation system with the control law \eqref{eq:compact_position_sign2} can be equivalently written as 
\begin{align}
\dot{p}_i=  \sum_{j\in \mathcal N_i} \frac{\text{sign}(||p_j - p_i|| - d_{ji})}{\|p_j - p_i\|} (p_j - p_i),
\end{align}
where $\mathcal N_i$ denotes agent $i$'s neighbor set, and $d_{ji}$ denotes the desired distance between agents $i$ and $j$. Denote $\mu_{ij}: = \frac{\text{sign}(||p_j - p_i|| - d_{ji})}{\|p_j - p_i\|}$. 
One can observe that  
\begin{align}
w_{ij}  &= \mu_{ij}, \,\,\,\text{if}\,\,\, (i,j) \in \mathcal{E}; \,\,\,\text{or}\,\,\,w_{ij}  = 0, \,\,\,\text{if}\,\,\, (i,j) \notin \mathcal{E}\nonumber \\  
w_{ii} &= -\sum_{j\in \mathcal N_i}   \mu_{ij}.
\end{align}
Therefore, the formation control system \eqref{eq:compact_position_sign2} can be rewritten in the form as $\dot p_i = \sum_{j =1}^{n} w_{ij}(t) p_j(t)$, which can be considered  as a special case of the general coupled system described by \eqref{eq:coupled_linear_wij}.  The subspace-preserving property then follows from \cite{sun2017dimensional}.

Now we define a weighted Laplacian matrix as $L_\omega = B \Omega B^\top$,  with the state-dependent diagonal weight matrix $\Omega$ defined as  $\Omega = \text{diag}(\omega_1, \omega_2, \ldots, \omega_{|\mathcal{E}|})$. Then a compact form of the formation position system can be obtained as 
\begin{align}
\dot p = - (L_\omega \otimes I_d)p.
\end{align}
Furthermore, from  \eqref{z_equation}
there holds
\begin{align}
\dot z &= {\overline B}^{\top} \dot p = - \left((B^\top B \Omega B^\top) \otimes I_d \right) p \nonumber \\
&= - \left((B^\top B \Omega) \otimes I_d \right) z.
\end{align}
\end{proof}
Then again, the affine-subspace-preserving property of the $z$ system in the sense of \eqref{eq:affinesubspacepreserving_sense} follows from \cite[Theorem 1]{sun2017dimensional}. \qedsymbol
\begin{remark}
An intuitive interpretation of the (affine) subspace-preserving property shown in 
Lemma~\ref{lemma_dimension_sign} is the invariance of collinear or coplanar positions for formation control systems in 2-D/3-D spaces. 
That is, for 2-D formations, if all    agents start with collinear positions, they will always be in that collinear subspace spanned by their initial positions under the general  control law described by \eqref{eq:compact_position_sign2}. Similarly, for 3-D formations, if all the agents start with coplanar (resp. collinear) positions, then their positions will always be coplanar (resp. collinear) under the control \eqref{eq:compact_position_sign2}. 
The paper \cite{liu2014controlling} presents a detailed analysis on such a collinear invariance principle for a triangular formation system under the control  \eqref{eq:compact_position_sign2}. In this sense, Lemma~\ref{lemma_dimension_sign} presents a more general invariance result which applies for any 2-D/3-D formation control systems with the control law  \eqref{eq:compact_position_sign2}. 
\end{remark}
}

Similar to the analysis in deriving the distance error system shown in Section \ref{sec:properties_system}, the distance error system with binary distance information can be obtained as
\begin{align} \label{error system_sign}
\dot e \in \mathcal{F}  [ -  D_{\tilde z}  R(z) R^{\top}(z) D_{\tilde z} \text{sign}(e)], \,\,\,\,\text{a. e.}
\end{align}
Again, similar to the analysis for \eqref{error system}, one can also show that \eqref{error system_sign} is a self-contained system when $e$ takes values locally around the origin.

\subsection{Convergence analysis}
The main result in this section is stated in the following   convergence theorem for the formation controller \eqref{eq:compact_position_sign2} with  binary distance information.
\begin{theorem}  \label{theorem:binary_distance}
Suppose the target formation is infinitesimally and minimally rigid, the initial formation shape is close to the target formation shape, and the formation controller \eqref{eq:position_sign2} with binary distance information is applied.
\begin{itemize}
\item The formation converges locally to a static target formation shape;
\item The convergence is achieved within a finite time upper bounded by $T^* = \frac{  \|e(0)\|_1}{ \bar \lambda_{\text{min}}}$ with $\bar \lambda_{\text{min}}$  defined as
{\color{black} 
\begin{align}  \label{eq:finite_time_sign}
\bar \lambda_{\text{min}}  =  \mathop {\min }\limits_{e \in \mathcal{B}} \lambda  (Q(e)) >0,
\end{align}
where $Q(e) :=D_{\tilde z} R(z)R^{\top}(z) D_{\tilde z}$, and $\mathcal{B}(\rho)= \{e: V(e) \leq \rho\}$ is a sub-level set of some suitably small constant $\rho>0$,    such that when $e \in \mathcal{B}(\rho)$ the  formation is infinitesimally minimally rigid.
}

\end{itemize}
\end{theorem}

\begin{proof}
Choose the Lyapunov function defined as $V = \sum_{k=1}^m V_k(e_k)$ with $V_k(e_k) = |e_k|$ for the distance error system \eqref{error system_sign}. Note that $V$ is a convex and regular function of $e$. Also $V$ is locally Lipschitz at $e = 0$ and is continuously differentiable at all other points.  The generalized derivative of $V_k(e_k)$ can be calculated as
\begin{align}
\partial V_k = \left\{
\begin{array}{cc}
       1,  &   e_{k} >0;    \\ \nonumber
       [-1, 1],  &   e_{k} = 0;   \\ \nonumber
       -1,  &   e_{k} < 0.
       \end{array}
      \right.
\end{align}
and the generalized derivative of $V$ can be calculated similarly via the product rule (see \cite{cortes2008discontinuous}).  We  define a sub-level  set $\mathcal{B}(\rho)= \{e: V(e) \leq \rho\}$ for some suitably small  $\rho$,    such that when $e \in \mathcal{B}(\rho)$ the  formation is infinitesimally minimally rigid and $R(z)R^{\top}(z)$ and $D_{\tilde z}$ are \emph{positive definite}. Now the matrix $Q(e) :=D_{\tilde z} R(z)R^{\top}(z) D_{\tilde z}$ is also \emph{positive definite} when $e \in \mathcal{B}(\rho)$. Let $\bar \lambda_{\text{min}}$ denote the smallest eigenvalue of $Q(e)$ when $e(p)$ is in the compact set $\mathcal{B} $ (i.e. $\bar \lambda_{\text{min}}  =  \mathop {\min }\limits_{e \in \mathcal{B}} \lambda  (Q(e)) >0$).

In the following, we calculate the set-valued derivative of $V$ along the trajectory described by the differential inclusion \eqref{error system_sign}.  The argument follows similarly to the analysis in the proof of Theorem \ref{theorem:convergence_quantization}.
By applying the calculation rule for the set-valued derivative (see \cite[Pages 62-63]{cortes2008discontinuous}), one can obtain
\begin{align}
\dot V(e)_{\eqref{error system_sign}} & \in \tilde{\mathcal{L}}_{\eqref{error system_sign}}V(e)  = \{  a \in \mathbb{R}| \exists v \in \dot e_{\eqref{error system_sign}}, \,\,\, \nonumber \\
 & \text{such that}\,\,\, \zeta^{\top} v= a, \forall \zeta \in \partial V(e) \}.
\end{align}
If the set $\tilde{\mathcal{L}}_{\eqref{error system_sign}}V(e)$   is not empty, there exists   $v \in -Q(e)  \text{sign}(e)$ such that $\zeta^{\top} v= a$ for all $\zeta \in \partial V(e)$. A natural choice of $v$ is to set $v \in -Q(e) \zeta$, with which one can obtain $a = -\text{sign}^{\top}(e)Q(e)\text{sign}(e)$.
 Then one can further show
\begin{align} \label{eq:Lya_derivative1}
\text{max}(\tilde{\mathcal{L}}_{\eqref{error system_sign}}V(e))   \leq    - \bar \lambda_{\text{min}} \text{sign}(e)^{\top} \text{sign}(e),
\end{align}
if the set is not empty, while if it is empty we adopt the convention  $\text{max}(\tilde{\mathcal{L}}_{\eqref{error system_sign}}V(e)) = - \infty$. Note that this implies that $V$ is non-increasing, and consequently the Filippov solution $e(t)$ is bounded. Thus, all solutions to \eqref{error system_sign} (as well as the solutions to \eqref{eq:compact_position_sign2}) are complete and can be extended to $t = \infty$ (i.e., there is no finite escape time). It can be seen that $\text{max}(\tilde{\mathcal{L}}_{\eqref{error system_sign}}V(e)) \leq 0$ for all $e \in \mathcal{B}(\rho)$ and
$0 \in \text{max}(\tilde{\mathcal{L}}_{\eqref{error system_sign}}V(e))$ if and only if $e  = 0$. According to the nonsmooth invariance principle \cite[Theorem 3]{bacciotti1999stability}, the asymptotic  convergence   is proved.

We then prove the stronger convergence result, i.e., the finite-time convergence. From the definition of the \emph{sign} function in \eqref{eq:definition_sign}, there holds $\text{sign}(e)^{\top} \text{sign}(e) >1$ for any $e \neq 0$, which implies
\begin{align} \label{eq:Lya_derivative1}
\text{max}(\tilde{\mathcal{L}}_{\eqref{error system_sign}}V(e))   \leq   - \bar \lambda_{\text{min}}
\end{align}
for any $e \neq 0$. Thus, by applying the Finite-time Lyapunov Theorem \cite{cortes2006finite}, any solution starting at $e(0) \in \mathcal{B}(\rho)$ reaches the origin in finite time, and the convergence time is upper bounded by $T^* = V(e(0))/\bar \lambda_{\text{min}} = \|e(0)\|_1/\bar \lambda_{\text{min}}$.

\end{proof}

\begin{remark}
(\textbf{Finite time formation convergence}) Different to the finite time convergence to an approximate formation shape under uniform quantizers as shown in Theorem \ref{theorem:convergence_quantization},  in Theorem \ref{theorem:binary_distance} it is shown the formation system converges locally to a correct formation shape  under binary distance measurements, which is a more desirable convergence result. Also,  compared to the finite time formation controller discussed in the paper \cite{sun2016finite} in which a $sig$ function is used, the finite time formation controller in \eqref{eq:position_sign2} requires less information in the distance measurements, in which very coarse  measurements in terms of binary signals are sufficient. 
\end{remark}

{\color{black} 
\begin{remark}
(\textbf{Effects on finite settling time}) According to the formula of the convergence time associated with \eqref{eq:finite_time_sign}, the upper bound on the finite setting time is affected by the initial shape (in the form of 1-norm of the distance error vectors) and the least singular value of the matrix $R^{\top}(z) D_{\tilde z}$ when the formation shape is evolved in the set $\mathcal{B}$. Roughly speaking, when a shape is close enough to the target shape, the least singular value $\bar \lambda_{\text{min}}$ could be approximated by $ \lambda  (Q(0))$, i.e., the least singular value when the formation is at the desired shape. In this sense, a formation control system  with a large $ \lambda  (Q(0))$ at the desired shape will generally have a shorter convergence time under the control law \eqref{eq:position_sign2} when the initial formation shape is close to the target shape. 
\end{remark}

}

\begin{remark} (\textbf{Dealing with chattering})
In the controller \eqref{eq:position_sign2} the sign function is used, which may cause chattering of the solutions to the formation system  when the formation is very close to the desired one (i.e. when $e$ is very close to the origin). This is because in practice imperfections (e.g., perturbations in measurements or delays) could cause agents' state trajectories to `chatter' across the discontinuity surface (see e.g. \cite[Chapter 3.5]{sastry2013nonlinear}).   Possible solutions to eliminate the chattering include the following:
\begin{itemize}
\item Add deadzone (approximated by smooth functions) to the sign function around the origin (similar to the case of uniform quantizers; see Part 1 of Theorem \ref{theorem:convergence_quantization}).
This will give rise to a trade-off in the convergence, i.e., the distance error does not converge to  the origin but to a bounded set, the size of which depends on  (for a fixed number of agents) how large the deadzone parameter is chosen  (see e.g. \cite{lee2007chattering, gupta2006frequency});
\item Use the hysteresis principle in the quantization function design \cite{ceragioli2011discontinuities}.

\end{itemize}
The adoption  of the above techniques to avoid chattering will be discussed in future research.
\end{remark}

\section{Asymmetric  uniform  quantizer} \label{sec:asymmetric_quantizer}
In \cite{liu2012quantization}, it has been shown that when an \emph{asymmetric} uniform  quantizer (defined below) is applied to double-integrator consensus dynamics some undesirable motions may occur. In this section we investigate whether there  are undesired motions for rigid formation control in the presence of an {asymmetric} uniform  quantizer.

We consider the following   \emph{asymmetric} uniform quantizer (the same as in \cite{liu2012quantization}), defined by
\begin{align} \label{eq:asy_uniform_quantizer}
q_u^*(x) = \delta_u \left( \floor*{\frac{x}{\delta_u}} \right ),
\end{align}
where $\delta_u$ is a positive number and $\floor*{a}, a \in \mathbb{R}$ denotes   \emph{the greatest
integer that is less than or equal to} $a$. For a comparison of the uniform quantizers defined in \eqref{eq:sym_uniform_quantizer} and in \eqref{eq:asy_uniform_quantizer}, see Fig. \ref{fig:quantization_function}.

\begin{figure}
  \centering
  \includegraphics[width=110mm]{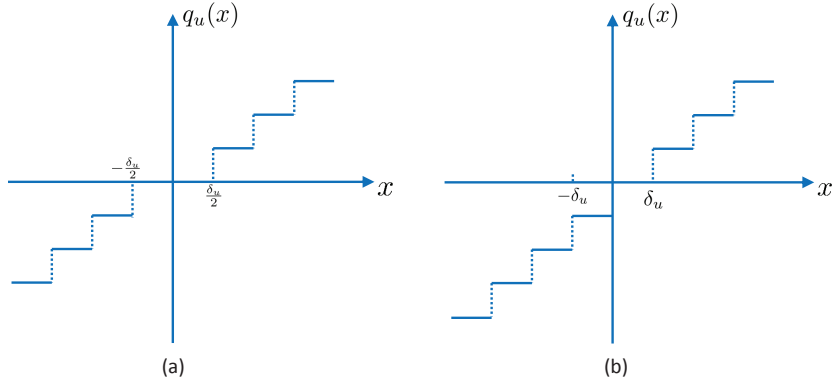}
  \caption{(a) Symmetric uniform quantizer function, defined in \eqref{eq:sym_uniform_quantizer}. (b) Asymmetric uniform quantizer function, defined in \eqref{eq:asy_uniform_quantizer}.}
  \label{fig:quantization_function}
\end{figure}

\subsection{Motivating example: two-agent formation case}
We first consider a two-agent formation case. Suppose two agents are controlled to achieve an inter-agent distance of $d_{12}$  with the  quantization function \eqref{eq:asy_uniform_quantizer}. The system dynamics for agents 1 and   2 can be described, respectively, as
\begin{align} \label{eq:asy_two_agent1}
\dot{p}_1=  q_u^* \left( ||z_1|| - d_{12} \right)\, \hat z_1
\end{align}
and
\begin{align} \label{eq:asy_two_agent2}
\dot{p}_2= - q_u^* \left( ||z_1|| - d_{12} \right)\, \hat z_1
\end{align}
where $z_1 = p_2 - p_1$, and $q_u^*(\cdot)$ denotes the asymmetric  uniform  quantizer  in \eqref{eq:asy_uniform_quantizer}.

\begin{lemma}
Consider the two-agent formation control system \eqref{eq:asy_two_agent1} and \eqref{eq:asy_two_agent2} with the asymmetric quantization function \eqref{eq:asy_uniform_quantizer}.
\begin{itemize}
\item If the initial distance between agents 1 and 2 is greater than   $d_{12} + \delta_u$, then the inter-agent distance $\|z\|$ will converge  to $d_{12} + \delta_u$ and the final formation will be stationary;
\item If the initial distance between agents 1 and 2 is smaller than the desired distance $d_{12}$, then the inter-agent distance $\|z\|$ will converge to the desired distance $d_{12}$ and the final formation will be stationary;
\item If the initial distance between agents 1 and 2 is between $d_{12}$ and $d_{12} + \delta_u$, then  both agents 1 and 2 remain  stationary and the inter-agent distance $\|z\|$ remains unchanged.
\end{itemize}
\end{lemma}
The proof is obvious and is omitted here as it can be inferred from previous proofs.
\begin{remark}
In the above example it can be seen that in the case of an asymmetric uniform  quantizer, there   exist no undesired motions, which is different to the result observed in \cite{liu2012quantization} which showed unbounded velocities. Apart from the difference in system dynamics under discussions, the key difference that leads to the distinct behaviors is that when the asymmetric   quantizer is applied to the consensus dynamics (which is to quantize a vector), there holds $\mathcal{F}[q_u^*(r_i - r_j)] + \mathcal{F}[q_u^*(r_j - r_i)] = -\delta_u$ when $r_j - r_i \neq k \delta_u$, where  $r_j - r_i$ denotes the relative position vector (see Section 5 of \cite{liu2012quantization}). Note that in the above formation controller, the quantization applies only  to the distance error term (i.e. $q_u^* \left(  \|p_2-p_1\|-d_{12}  \right)$) which is a scalar, and the asymmetric property of the quantizer only affects the convergence of the distance term.
\end{remark}

\subsection{General formation case}
We consider the general formation case with more than two agents, in which each   agent  employs asymmetric uniform  quantizers in individual controllers.

\begin{theorem} \label{theorem:asym_case}
Suppose each individual agent takes the asymmetric uniform  quantizer \eqref{eq:asy_uniform_quantizer} in the quantized formation controller \eqref{eq:position_system}. Then the inter-agent distances converge within a finite time to the set $$F_{\text{aym}} = \{e | e_{k} \in   [0, \delta_u], k \in \{1, \dots, |\mathcal{E}|\}\}.$$
\end{theorem}
The proof is omitted here as it can be directly inferred from the previous proof of Theorem \ref{theorem:convergence_quantization}.

\section{Illustrative examples and simulations} \label{sec:quantization_simulations}
In this section we show several numerical examples  to illustrate the   theoretical results obtained in previous sections. In the following illustrative examples we consider the stabilization control of a five-agent minimally rigid formation in the 3-D space, as a continuation of Example \ref{exp:3D_formation}. The underlying graph describes  a double tetrahedron shape of nine edges (see Figure \ref{fig:tetrahedron_shape} for an illustration),   and the desired distances for all edges are   set as $6$. \footnote{Note that the realization of a target formation with the given nine desired distances is not unique up to rotation and translation \cite{hendrickson1992conditions}.} The initial positions are chosen such that the initial formation is infinitesimally rigid and is close to a target formation shape. For all simulations, we set the quantization gain as $\delta_u = 0.5$.

Agents trajectories, the final formation shape and the evolutions of nine distance errors under symmetric uniform quantization and under logarithmic quantization are shown in Figure  \ref{fig:sym_uniform} and Figure \ref{fig:Loga}, respectively. It is obvious from these two figures that with symmetric uniform quantizer the formation errors converge to the bounded set
$F_{\text{approx}} = \{e | e_{k} \in [-0.25, 0.25], k \in \{1, \dots, m\}\}$ in a \textit{finite time}, and with the logarithmic quantizer the formation converges to the target shape asymptotically, which are consistent with the theoretical results in Theorem \ref{theorem:convergence_quantization}.

The formation convergence behavior with binary distance measurements under the quantization strategy
\eqref{eq:position_sign2} is depicted in Figure \ref{fig:sign}. It can be seen from Figure \ref{fig:sign} that with very coarsely quantized distance measurement via a simple signum function as in  \eqref{eq:position_sign2}, the formation converges to the target shape within a finite time, but the price to be paid is   the occurrence of chattering (as shown in the right part of Figure \ref{fig:sign}).

Finally, when the asymmetric uniform quantizer \eqref{eq:asy_uniform_quantizer} is used in the formation control system \eqref{eq:position_system}, the formation converges to an approximate one with all distance errors converging to the bounded set $F_{\text{aym}} = \{e | e_{k} \in   [0, 0.5], k \in \{1, \dots, |\mathcal{E}|\}\}$ within a finite time, as shown in  Figure \ref{fig:asym}. This supports the conclusion of Theorem \ref{theorem:asym_case}.

{\color{black} For a comparison of formation convergences with different quantization functions, we plot the trajectories of the Lyapunov functions (which are functions of distance errors) for the formation system under the above four different quantization strategies, as shown in Fig.~\ref{fig:comparisons}. In the simulations the initial conditions were chosen to be the same for the four cases. It can be observed from  Fig.~\ref{fig:comparisons} that under symmetric uniform quantization, binary quantization, and asymmetric uniform quantization, the convergences of Lyapunov functions (as functions of distance errors) are achieved within a finite time (though the final values depend on different quantization functions), and that under logarithmic quantization the distance errors converge to zero asymptotically.  All these results are consistent  with the above simulations and support the theoretical developments in the previous sections.  }

\begin{figure}
  \centering
  \includegraphics[width=125mm]{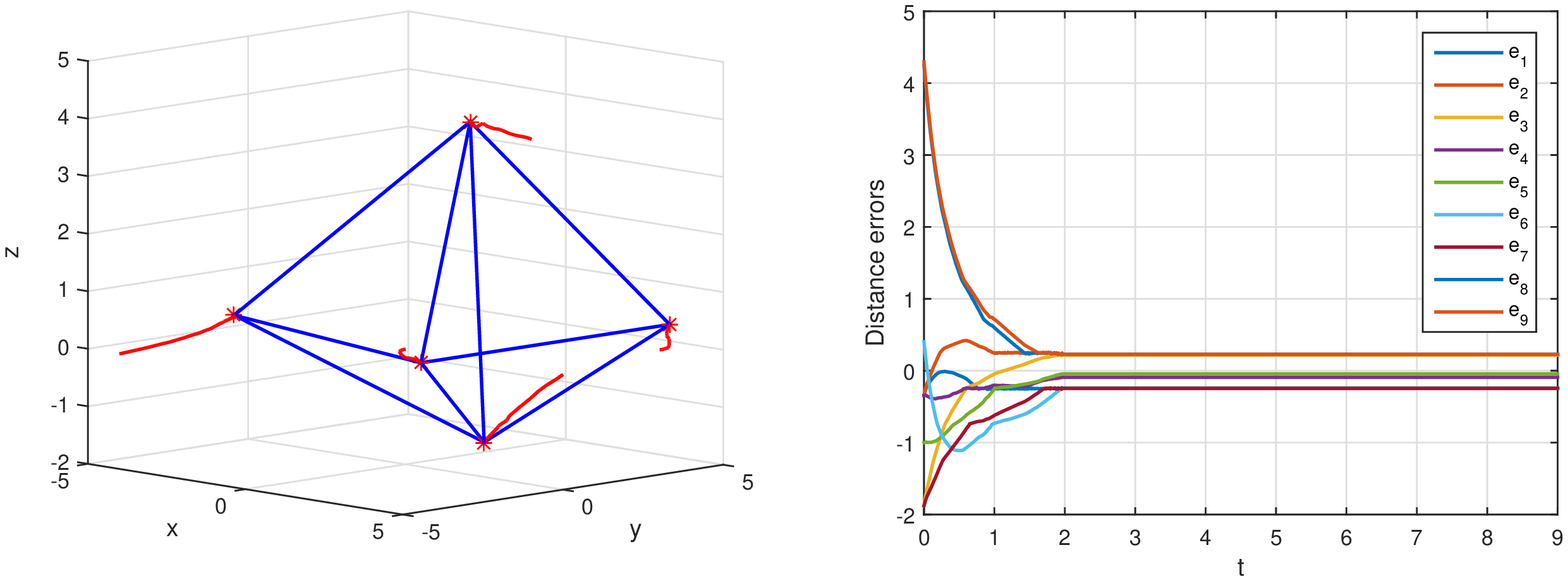}
  \caption{\small Stabilization control of a rigid formation:  symmetric uniform quantization case.
  Left: the trajectories of five agents and the final formation shape. Right: Time evolutions of the  distance errors.  It is obvious from the right figure that the formation shape converges to an approximately correct one in a finite time.}
  \label{fig:sym_uniform}
\end{figure}

\begin{figure}
  \centering
  \includegraphics[width=125mm]{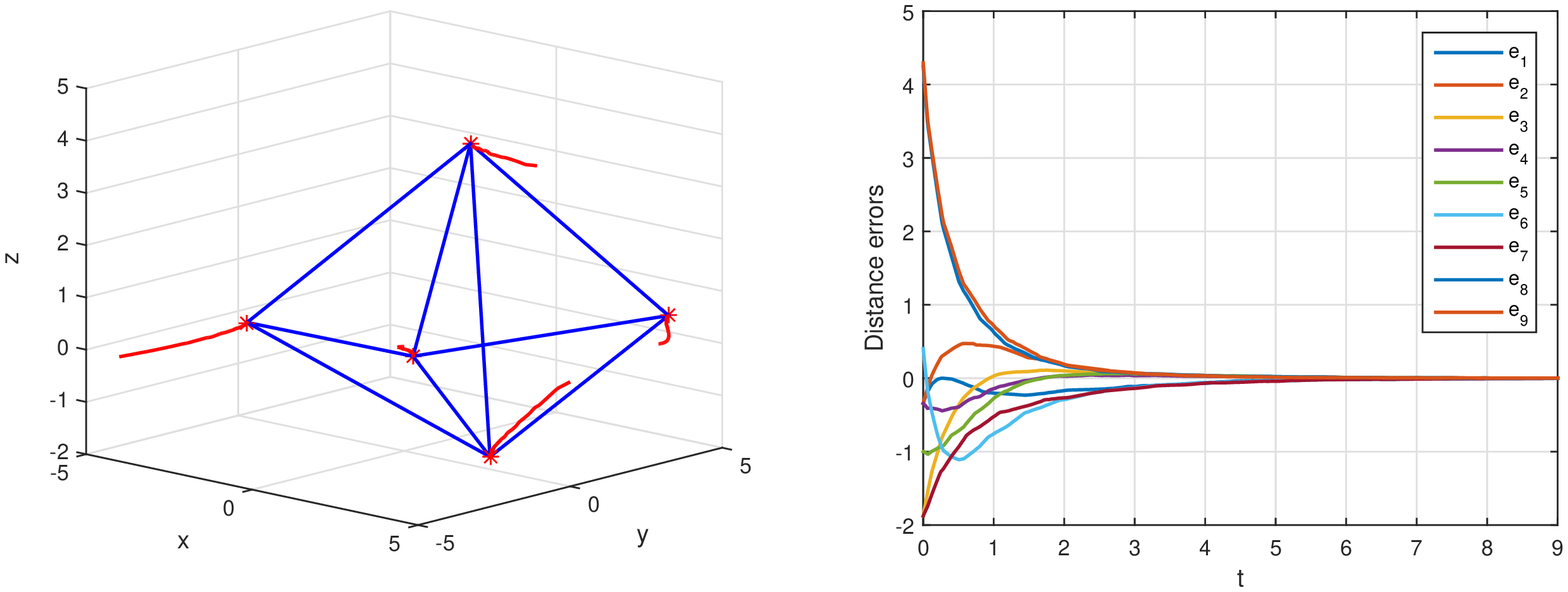}
  \caption{\small Stabilization control of a rigid formation:  logarithmic quantization case.
  Left: the trajectories of five agents and the final formation shape. Right: Time evolutions of the  distance errors. }
  \label{fig:Loga}
\end{figure}

\begin{figure}
  \centering
  \includegraphics[width=125mm]{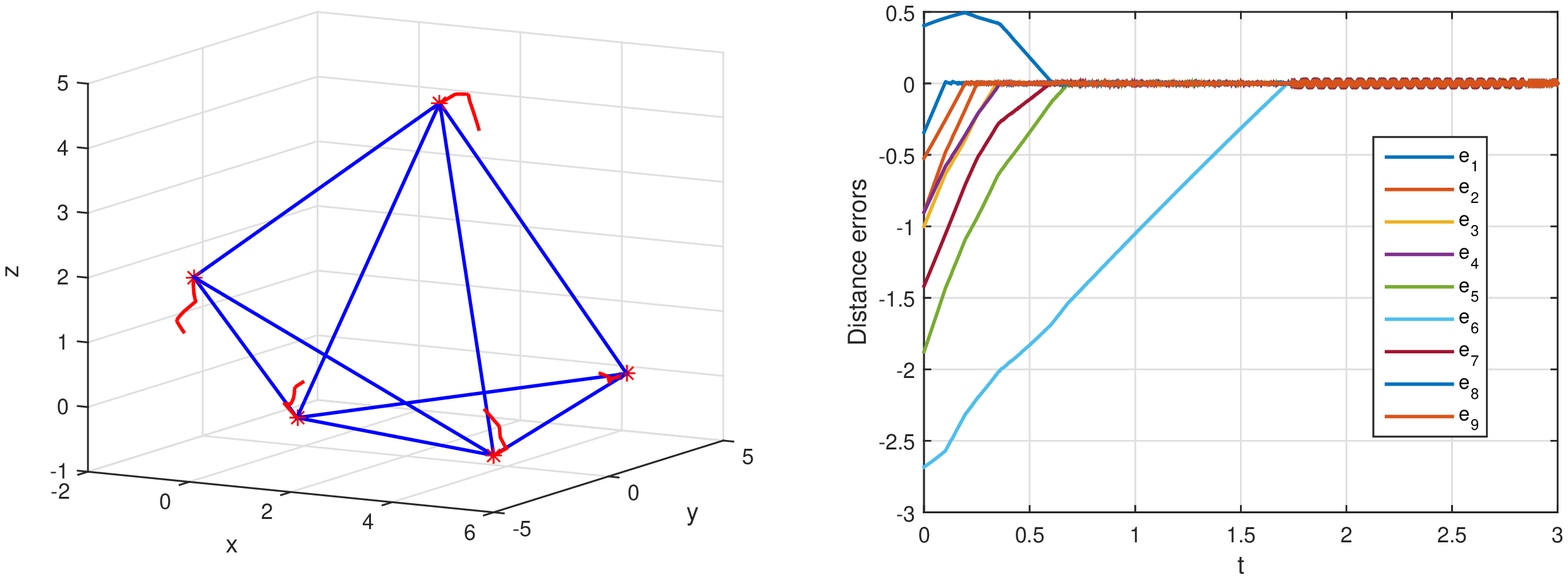}
  \caption{\small Stabilization control of a rigid formation:  binary measurement case.
  Left: the trajectories of five agents and the final formation shape. Right: Time evolutions of the  distance errors.}
  \label{fig:sign}
\end{figure}

\begin{figure}
  \centering
  \includegraphics[width=115mm]{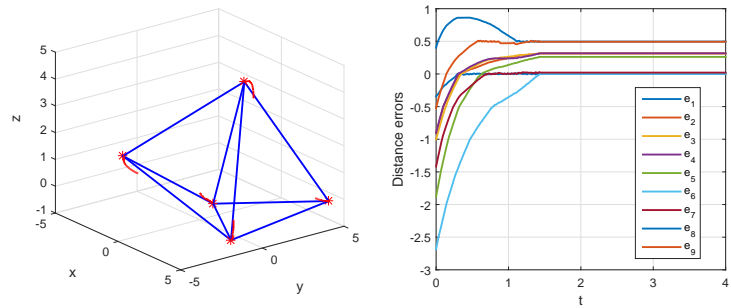}
  \caption{\small Stabilization control of a rigid formation:  asymmetric uniform quantization case.
  Left: the trajectories of five agents and the final formation shape. Right: Time evolutions of the  distance errors.}
  \label{fig:asym}
\end{figure}

\begin{figure}
  \centering
  \includegraphics[width=115mm]{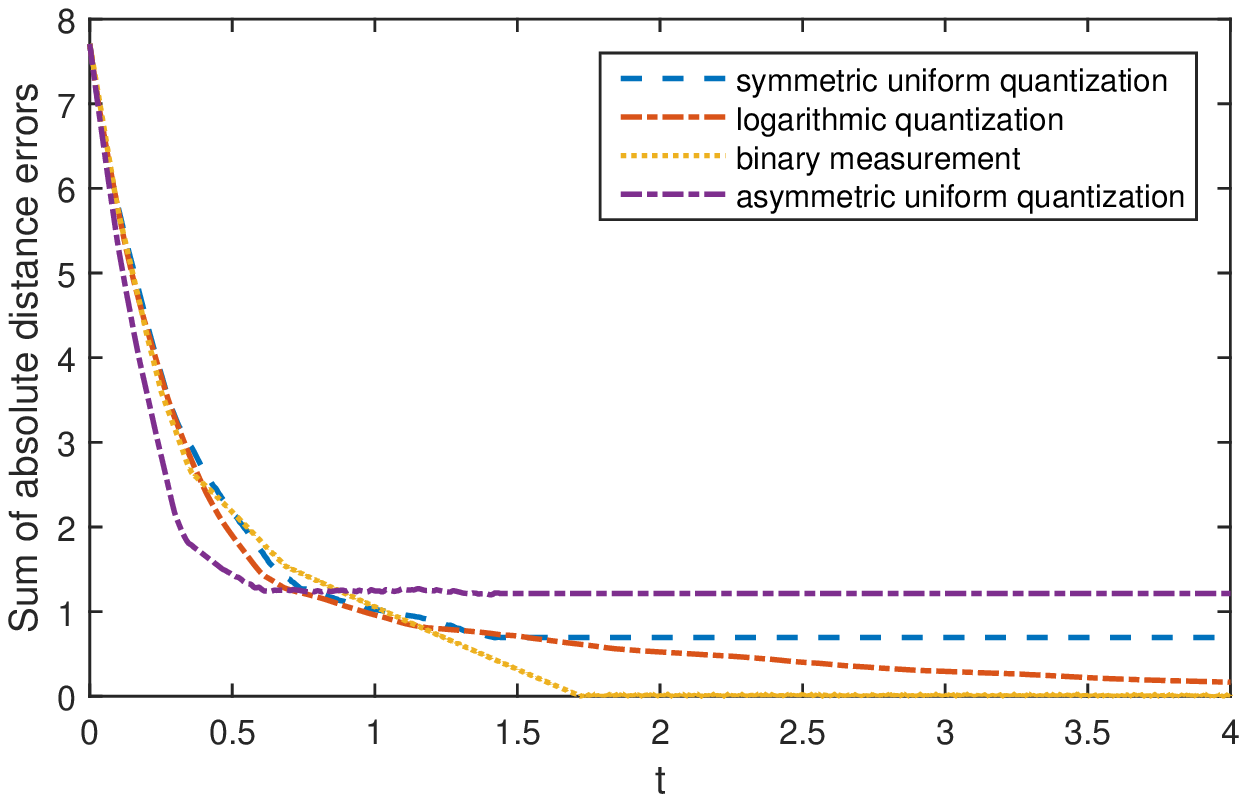}
  \caption{\small Comparisons of convergences of Lyapunov functions (as functions of formation distance errors) under different quantization functions. 
  }
  \label{fig:comparisons}
\end{figure}

\section{Concluding remarks} \label{sec:Conclusions}
In this paper we consider the rigid formation control problem with quantized distance  measurements. We have discussed in detail the quantization effect on the convergence of rigid formation shapes under two commonly-used quantizers. In the case of the symmetric uniform quantizer, all distances will converge locally to a bounded set, the size of which depends on the quantization error. In the case of the logarithmic quantizer, all distances converge locally to the desired values. We also consider a special quantizer with a signum function, which allows each agent to use very coarse distance measurements (i.e. binary information on whether it is close or far away to neighboring agents with respect to the desired distances). We show in this case the formation shape can still be achieved within a finite time. We further discuss the case of an asymmetric quantizer applied in rigid formation control system, and  analyze the convergence property of distance errors.

\section*{Acknowledgment}
This work was supported by the Australian Research Council under grant DP130103610 and DP160104500. Z. Sun was supported by the Prime Minister's Australia Asia Incoming Endeavour Postgraduate Award. H. Garcia de Marina was supported by EU H2020 Mistrale project under grant agreement no. 641606. The work of M. Cao was supported in part by the Netherlands Organization for Scientific Research (NWO-vidi-14134).

\bibliography{Quantization_formation}
\bibliographystyle{wileyj}

\end{document}